\documentclass[a4paper,UKenglish,final]{lipics}
\pdfoutput=1
\usepackage{amssymb}
\usepackage{amstext}
\usepackage{stmaryrd}
\usepackage{url}
\usepackage{tikz}
\usetikzlibrary{arrows,automata,backgrounds,calc,chains,fadings,folding,decorations}
\usetikzlibrary{fit,patterns,positioning,mindmap}
\usetikzlibrary{shadows,shapes,through,trees}

\title{Termination of LCTRSs\footnote{The research described in this
paper is supported by the Austrian Science Fund (FWF) international
project I963 and the Japan Society for the Promotion of Science.}}

\author[1]{Cynthia Kop}
\affil[1]{Department of Computer Science, University of Innsbruck\\
  Technikerstra\ss e 21a, 6020 Innsbruck, Austria \\
  \texttt{Cynthia.Kop@uibk.ac.at}}
\authorrunning{C. Kop}
\Copyright[by]{Cynthia Kop}

\newcommand{\setvars}{\mathcal{V}}
\newcommand{\arrtype}{\Rightarrow}
\newcommand{\arrfunc}{\Longrightarrow}
\newcommand{\T}{\mathcal{I}}
\newcommand{\Terms}{\mathcal{T}\mathit{erms}}
\newcommand{\J}{\mathcal{J}}
\newcommand{\interpret}[1]{\llbracket #1 \rrbracket}
\newcommand{\Bool}{\mathbb{B}}

\newcommand{\arrz}{\rightarrow}
\newcommand{\arrzrule}{\arrz_{\mathtt{rule}}}
\newcommand{\arrzcalc}{\arrz_{\mathtt{calc}}}

\newcommand{\arrr}[1]{\arrz_{#1}^*}

\newcommand{\constraint}[1]{[#1]}

\newcommand{\FV}{\mathit{Var}}

\newcommand{\Rules}{\mathcal{R}}
\newcommand{\subst}[2]{#1#2}

\newcommand{\Sigmaterms}{\Sigma_{\mathit{terms}}}
\newcommand{\Sigmalogic}{\Sigma_{\mathit{theory}}}

\newcommand{\Values}{\mathcal{V}al}

\newcommand{\up}[1]{#1^\sharp}
\newcommand{\DP}{\mathsf{DP}}

\renewcommand{\P}{\mathcal{P}}

\newcommand{\project}{\overline{\nu}}

\newcommand{\LVar}{\mathit{LVar}}
\renewcommand{\equiv}{=}
\newcommand{\toint}{\mu}

\newcommand{\afun}{f}
\newcommand{\bfun}{g}

\newcommand{\avar}{x}
\newcommand{\bvar}{y}

\newcommand{\aterm}{s}
\newcommand{\bterm}{t}
\newcommand{\cterm}{q}

\newcommand{\asort}{\iota}
\newcommand{\bsort}{\kappa}

\newcommand{\symb}[1]{\mathsf{#1}}

\newcommand{\Z}{\mathbb{Z}}

\newcommand{\nul}{\symb{0}}
\newcommand{\one}{\symb{1}}

\newcommand{\summ}{\symb{sum}}

\newcommand{\Int}{\symb{int}}
\newcommand{\BOOL}{\symb{bool}}

\newcommand{\ack}{\symb{A}}

\newcommand{\standard}[1]{}
\newcommand{\nonstandard}[1]{#1}

\def\copyrightline{}
\begin{document}

\maketitle

\begin{abstract}
\emph{Logically Constrained Term Rewriting Systems} (LCTRSs) provide a
general framework for term rewriting with constraints.  
We discuss a simple dependency pair approach to prove termination of
LCTRSs.  We see that existing techniques transfer to the constrained
setting in a natural way.
\end{abstract}

\section{Introduction}

In~\cite{kop:nis:13}, \emph{logically constrained term rewriting
systems} are introduced (building on~\cite{fur:nis:sak:kus:sak:08}
and~\cite{fal:kap:09}).  These \emph{LCTRSs} combine many-sorted
term rewriting with constraints in an arbitrary theory, and can be
used for analysing for instance imperative programs.

\emph{Termination} is an important part of such analysis, both for
its own sake (to guarantee finite program evaluation), and to create
an induction principle that can be used as part of other analyses
(for instance proofs of confluence~\cite{ter:03} or function
equality~\cite{fur:nis:sak:kus:sak:08}).

In unconstrained term rewriting,
many termination techniques exist, often centred around
\emph{dependency pairs}~\cite{art:gie:00}.
Some of these methods have
also been transposed to integer rewriting with
constraints~\cite{fal:kap:09}.  However, that setting is focused
purely on proving termination for its own sake, and thus poses very
strong restrictions on term and rule formation.

In this paper, we will see how a basic dependency pair approach can
be defined for LCTRSs, and extend several termination methods which
build around dependency pairs.

\section{Preliminaries (from~\cite{kop:nis:13})}

We assume standard notions of many-sorted term rewriting to be
well-uderstood.

Let $\setvars$ be an infinite set of sorted variables, $\Sigma =
\Sigmaterms \cup \Sigmalogic$ be a many-sorted signature, $\T$ a
mapping which assigns to each sort occurring in $\Sigmalogic$ a
set, and $\J$ a function which maps each $\afun : [\asort_1 \times
\ldots \times \asort_n] \arrtype \bsort \in \Sigmalogic$ to a
function $\J_\afun$ in $\T_{\asort_1} \arrfunc \ldots \arrfunc
\T_{\asort_n} \arrfunc \T_\bsort$.
For every sort $\asort$ occurring in $\Sigmalogic$ we also fix a set
$\Values_\asort \subseteq
\Sigmalogic$ of \emph{values}: function symbols $a : [] \arrtype
\asort$, where $\J$ gives a one-to-one mapping from $\Values_\asort$
to $\T_\asort$.
A value $c$ is identified with the term $c()$.
The elements of $\Sigmalogic$ and $\Sigmaterms$ overlap only on
values.

We call a term in $\Terms(\Sigmalogic,\setvars)$ a
\emph{logical term}.
For ground logical terms, we define
$
\interpret{\afun(\aterm_1,\ldots,\aterm_n)} :=
  \J_\afun(\interpret{\aterm_1},\ldots,\interpret{\aterm_n})
$.
A ground logical term $\aterm$ \emph{has value} $\bterm$ if $\bterm$
is a value such that $\interpret{\aterm} = \interpret{\bterm}$.
Every ground logical term has a unique value.
A \emph{constraint} is a logical term of some sort $\BOOL$
with $\T_\BOOL = \Bool$, the set of booleans.
A constraint $\aterm$ is \emph{valid} if
$\interpret{\subst{\aterm}{\gamma}}_\J = \top$
for all substitutions $\gamma$ which map the variables in
$\FV(\aterm)$ to a value.

A \emph{rule} is a triple $l \arrz r\ \constraint{\varphi}$ where $l$
and $r$ are terms with the same sort and $\varphi$ is a constraint;
$l$ is not a logical term (so also not a variable).
If $\varphi = \symb{true}$
with $\J(\symb{true}) = \top$, the rule is \linebreak just denoted $l
\arrz r$.
We define $\LVar(l \arrz r\ \constraint{\varphi}))$ as $\FV(\varphi)
\cup (\FV(r) \setminus \FV(l))$.
A substitution $\gamma$ \emph{respects} $l \arrz r\ 
\constraint{\varphi}$ if
$\gamma(\avar)$ is a value for all $\avar \in
\LVar(l \arrz r\ \constraint{\varphi})$ and $\subst{\varphi}{\gamma}$
is valid.

Given a set of rules $\Rules$, the \emph{rewrite relation}
$\arrz_{\Rules}$ is the union of $\arrzrule$ and $\arrzcalc$, where:
\begin{itemize}
\item $C[l\gamma] \arrzrule C[r\gamma]$ if $l \arrz r\ 
  \constraint{\varphi} \in \Rules$ and $\gamma$ respects $l \arrz r\ 
  \constraint{\varphi}$;
\item $C[\afun(\aterm_1,\ldots,\aterm_n)] \arrzcalc C[v]$ if $\afun
  \in \Sigmalogic\setminus\Sigmaterms$, all\ $\aterm_i$ values and
  $v$ is the value of $\afun(\vec{\aterm})$
\end{itemize}

A reduction step with $\arrzcalc$ is called a \emph{calculation}.
In an LCTRS with rules $\Rules$, the \emph{defined symbols} are all
symbols $\afun$ such that a rule $\afun(\vec{l}) \arrz r\ 
\constraint{\varphi}$ exists in $\Rules$.  Symbols $\afun \in
\Sigmalogic \setminus \Values$ are called \emph{calculation symbols}
and all other symbols are \emph{constructors}.
%

\begin{example}\label{ex:ackermann}
We consider an LCTRS with sorts $\Int$ and $\BOOL$, with $\T_\BOOL =
\Bool$ and $\Int$ mapped to the set of 16-bit signed
integers; addition is sensitive to overflow.
The rules are a naive implementation of the Ackermann function
(which will likely fall prey to overflows):
\vspace{-5pt}
\[
\begin{array}{rcllrcl}
\ack(m,n) & \arrz & \ack(m-\one,\ack(m,n-\one)) &
  \constraint{m \neq \nul \wedge n \neq \nul} &
\quad
\ack(\nul,n) & \arrz & n+\one \\
\ack(m,\nul) & \arrz & \ack(m-\one, \one) &
  \constraint{m \neq \nul} \\
\end{array}
\]
\vspace{-5pt}
$\ack$ is a defined symbols, $+,-,\neq,\wedge$ calculation symbols,
and all integers are constructors.
\end{example}


\section{Dependency Pairs}

As the basis for termination analysis, we will consider
\emph{dependency pairs}~\cite{art:gie:00}.
We first introduce a fresh sort $\symb{dpsort}$, and for all defined 
symbols $\afun : [\asort_1 \times \ldots \times \asort_n] \arrtype
\bsort$ also a new symbol $\up{\afun} : [\asort_1 \times \ldots
\times \asort_n] \arrtype \symb{dpsort}$.  If $\aterm =
\afun(\aterm_1,\ldots, \aterm_n)$ with $\afun$ defined,
then $\up{\aterm} := \up{\afun}(\aterm_1,\ldots,\aterm_n)$.

The dependency pairs of a given rule $l \arrz r\ \constraint{\varphi}
$ are all rules of the form $\up{l} \arrz \up{p}\ \constraint{\varphi
}$ where $p$ is a subterm of $r$ which is headed by a defined symbol.
The set of dependency pairs for a given set of rules $\Rules$,
notation $\DP(\Rules)$, consists of all dependency pairs of any rule
in $\Rules$.

\begin{example}\label{ex:ackdp}
Noting that for instance $\up{\ack}(m,\nul) \arrz m \up{-} \one$ is
\emph{not} a dependency pair, since $-$ is a calculation symbol
and not a defined symbol, Example~\ref{ex:ackermann} has three
dependency pairs:
\vspace{-5pt}
\[
\begin{array}{lrcll}
1. & \up{\ack}(m,\nul) & \arrz & \up{\ack}(m-\one,\one) &
  \constraint{m \neq \nul} \\
2. & \up{\ack}(m,n) & \arrz & \up{\ack}(m-\one,\ack(m,n-\one)) &
  \constraint{m \neq \nul \wedge n \neq \nul} \\
3. & \up{\ack}(m,n) & \arrz & \up{\ack}(m,n-\one) &
  \constraint{m \neq \nul \wedge n \neq \nul} \\
\end{array}
\]
\end{example}

Fixing a set $\Rules$ of rules, and given a set $\P$ of dependency
pairs, a \emph{$\P$-chain} is a sequence $\rho_1,\rho_2,\ldots$ of
dependency pairs such that 
  all $\rho_i$ are elements of $\P$, but with distinctly renamed
  variables
, and
there is 
some
$\gamma$ which respects all $\rho_i$,
  such that for all $i$: if $\rho_i = l_i \arrz p_i\ \constraint{
  \varphi_i}$ and $\rho_{i+1} = l_{i+1} \arrz p_{i+1}\ \constraint{
  \varphi_{i+1}}$, then $\subst{p_i}{\gamma} \arrr{\Rules} \subst{l_{
  i+1}}{\gamma}$.
  Also, the strict subterms of $l_i\gamma$ terminate.
%
We call $\P$ a \emph{DP problem} and say that $\P$
is \emph{chain-free} if there is no infinite 
$\P$-chain.\footnote{In the literature, we consider tuples of sets
and flags, which is necessary if we also want to consider non-minimal
chains, innermost termination or non-termination.  For simplicity
those are omitted here.}\footnote{In the literature, the word
\emph{finite} is used instead of \emph{chain-free}.  Since we have a
single set instead of a tuple, we used a different word to avoid
confusion (as ``finite'' might refer to the number of elements).}

\begin{theorem}\label{thm:dpchain}
An LCTRS $\Rules$ is terminating if and only if $\DP(\Rules)$ is
chain-free.
\end{theorem}


\section{The Dependency Graph}

To prove chain-freeness of a DP problem, we might for instance use the
dependency graph:

\begin{definition}
A \emph{dependency graph approximation} of a DP problem $\P$ is a
graph $G$ whose nodes are the elements of $\P$ and which has an edge
between $\rho_1$ and $\rho_2$ if $(\rho_1,\rho_2')$ is a
$\P$-chain, where $\rho_2'$ is a copy of $\rho_2$ with fresh
variables.  $G$ may have additional edges.
\end{definition}

\begin{theorem}\label{thm:graph}
A DP problem $\P$ with graph approximation $G$ is chain-free if and
only if $\P'$ is chain-free for every strongly connected component
(SCC) $\P'$\! of $G$.
\pagebreak
\end{theorem}


\begin{example}\label{ex:graph}
Consider an LCTRS with rules $\Rules = \{ \afun(\avar) \arrz
\afun(\nul-\avar)\ \constraint{\avar > \nul} \}$.  Then
$\DP(\Rules) = \{ \up{\afun}(\avar) \arrz \up{\afun}(-\avar)\ 
\constraint{\avar > \nul} \}$.  The dependency graph of $\DP(
\Rules)$ has
one node, and no edges, since there is no substitution $\gamma$
which satisfies both $\gamma(\avar) > \nul$ and $\gamma(\bvar) >
\nul$ and yet has $(-\avar)\gamma \arrr{\Rules} \bvar\gamma$ (as
logical terms reduce only with $\arrzcalc$).
Thus, clearly every SCC of this graph is terminating, so
$\DP(\Rules)$ is chain-free, so $\Rules$ is terminating!
\end{example}

Of course, manually choosing a graph approximation is one thing,
but finding a good one \emph{automatically} is more difficult.  We
consider one way to choose such an approximation:

Given a DP problem $\P$, let $G_\P$ be the graph with the elements of
$\P$ as nodes, and with an edge from $l_1 \arrz r_1\ \constraint{
\varphi_1}$ to $l_2 \arrz r_2\ \constraint{\varphi_2}$ if the formula
$\varphi_1 \wedge \varphi_2' \wedge \psi(r_1,l_2',\LVar(l_1 \arrz r_1\ 
\constraint{\varphi_1})\linebreak \cup \LVar(l_2' \arrz r_2'\ 
\constraint{\varphi_2'}))$ is satisfiable (or its satisfiability
cannot be determined).  Here, $l_2' \arrz r_2'\ 
\constraint{\varphi_2'}$ is a copy of $l_2 \arrz r_2\ \constraint{
\varphi_2}$ with fresh variables,
and $\psi(\aterm,\bterm,L)$ is given by the 
clauses:

\begin{itemize}
\item $\psi(\aterm,\bterm,L) = \top$ if either $\aterm$ is a variable
  not in $L$, or $\aterm = \afun(\aterm_1,\ldots,\aterm_n)$ and one
  of: 
  \begin{itemize}
  \item $\afun$ is a defined symbol, and $\aterm \notin
    \Terms(\Sigmalogic,L)$,
  \item $\afun$ is a calculation symbol, $\bterm$ a value or
    variable, and $\aterm \notin \Terms(\Sigmalogic,L)$,
  \item $\afun$ is a constructor and $\bterm$ a variable not in $L$;
  \end{itemize}
\item $\psi(\aterm,\bterm,L) = \bigwedge_{i = 1}^n \psi(
  \aterm_i,\bterm_i,L)$ if $\aterm = \afun(\aterm_1,\ldots,
  \aterm_n)$ and $\bterm = \afun(\bterm_1,\ldots,\bterm_n)$ and
  $\afun$ not defined;
\item $\psi(\aterm,\bterm,L)$ is the formula $\aterm =
  \bterm$ if
  $\aterm \in \Terms(\Sigmalogic,L)$,\ 
  $\bterm \in \Terms(\Sigmalogic,\setvars)$ and
  $\aterm$ and $\bterm$ are not headed by the same theory symbol
  (we already covered that case);
\item $\psi(\aterm,\bterm,L) = \bot$ in all other cases.
\end{itemize}

\begin{theorem}\label{thm:graphapprox}
$G_\P$ is a graph approximation for $\P$.
\end{theorem}


This graph result and the given approximation
correspond largely with the result of~\cite{sak:nis:sak:11}.

\begin{example}
The graph in Example~\ref{ex:graph} is calculated with this method:
$\psi(\up{\afun}(-\avar),\up{\afun}(\bvar),\{\avar,\linebreak
\bvar\})
\wedge \avar > \nul \wedge \bvar > \nul$ evaluates to
$-\avar = \bvar \wedge \avar > \nul \wedge \bvar > \nul$ (as
$\up{\afun}$ is a constructor with respect to $\Rules$), which is
not satisfiable (as any decent SMT-solver over the integers can
tell us).
\end{example}

%

\section{The Value Criterion}

To quickly handle DP problems, we consider a technique similar to the
subterm criterion in the unconstrained case.  This \emph{value
criterion} can also be seen as a simpler version of polynomial
interpretations, which does not require ordering rules (see
Section~\ref{sec:redpair}).

\begin{definition}
Fixing a set $\P$ of dependency pairs, a \emph{projection function}
for $\P$ is a function $\nu$ which assigns to each symbol $\up{\afun}
: [\asort_1 \times \ldots \times \asort_n] \arrtype \symb{dpsort}$
a number $\nu(\up{\afun}) \in \{1,\ldots,n\}$.  A projection function
is extended to a function on terms as follows:
$\overline{\nu}(\up{\afun}(\aterm_1,\ldots,\aterm_n)) =
\aterm_{\nu(\up{\afun})}$.
\end{definition}

\begin{theorem}\label{thm:valuecrit}
Let $\P$ be a set of dependency pairs, $\asort$ a sort and $\nu$ a
projection function for $\P$, with the following property: for any
dependency pair $l \arrz r\ \constraint{\varphi} \in \P$, if
$\overline{\nu}(l)$ has sort $\asort$ and is a \standard{value or
variable}\nonstandard{logical term (this includes variables)}, then
the same holds for $\overline{\nu}(r)$.
Let moreover $\succ$ be a well-founded ordering relation on
$\T_\asort$ and $\succeq$ a quasi-ordering such that $\succ \cdot
\succeq\ \subseteq\ \succ$.
Suppose additionally that we can write $\P = \P_1 \cup \P_2$, such
that for all $\rho \equiv l \arrz r\ \constraint{\varphi} \in \P$:
\begin{itemize}
\item if $\project(l)$ is a logical term of sort $\asort$, then so is
  $\project(r)$, and $\FV(\project(r)) \subseteq \FV(\project(l))$;
\item if $\rho \in \P_1$, then $\project(l)$ has sort $\asort$
  and $\project(l) \in \Terms(\Sigmalogic,\LVar(\rho))$;
\item if $\overline{\nu}(l)$ has sort $\asort$ and $\project(l) \in
  \Terms(\Sigmalogic,\setvars)$, then $\varphi \Rightarrow
  \project(l) \succ \project(r)$ is valid if $\rho \in \P_1$, and
  $\varphi \Rightarrow \project(l) \succeq \project(r)$ is valid if
  $\rho \in \P_2$.
\end{itemize}
Then $\P$ is chain-free if and only if $\P_2$ is chain-free.
\end{theorem}

\begin{proof}
A chain with infinitely many elements of $\P_1$ gives an
infinite $\succeq^* \cdot \succ$ reduction.
\end{proof}

\begin{example}
Using the value criterion, we can complete termination analysis of
the Ackermann example.  Choosing for $\succ$ the \emph{unsigned}
comparison on bitvectors (so $n \succ m$ if either $n$ is negative
and $m$ is not, or $\mathit{sign}(n) = \mathit{sign}(m)$ and $n >
m$), and $\nu(\ack) = 1$, we have:
\begin{itemize}
\item $\up{\ack}(m,\nul) \arrz \up{\ack}(m-\one,\one)\ 
  \constraint{m \neq \nul}$:
  $
  (m \neq \nul) \Rightarrow m \succ m-1$
\item $\up{\ack}(m,n) \arrz \up{\ack}(m-\one,\ack(m,n-\one))\ 
  \constraint{m \neq \nul \wedge n \neq \nul}$:
  $
  (m \neq \nul \wedge n \neq \nul) \Rightarrow m \succ m-1$
\item $\up{\ack}(m,n) \arrz \up{\ack}(m,n-\one)\ 
  \constraint{m \neq \nul \wedge n \neq \nul}$
  $(m \neq \nul \wedge n \neq \nul) \Rightarrow m \succeq m$
\end{itemize}
All three are valid, so $\P$ is chain-free if $\P_2 =
\{ \up{\ack}(m,n) \arrz \up{\ack}(m,n-\one)\ \constraint{m \neq \nul
\wedge \nul \wedge n \neq \nul} \}$ is.
This we prove with another application of the value criterion, now
taking $\nu(\up{\ack}) = 2$.
\end{example}

Note that the difficulty to apply the value criterion is in finding a
suitable value ordering.  There are various systematic techniques for
doing this (depending on the underlying theory), but their specifics
are beyond the scope of this paper.

\section{Reduction Pairs}\label{sec:redpair}

Finally, the most common method to prove chain-freeness is the use of
a \emph{reduction pair}.

A reduction pair $(\succsim,\succ)$ is a
pair of a \emph{monotonic quasi-ordering} and a \emph{well-founded
partial ordering} on terms such that $\aterm \succ \bterm \succsim
\cterm$ implies $\aterm \succ \cterm$.  Note that it is not required
that $\succ$ is included in $\succsim$; $\succsim$ might also for
instance be an equivalence relation.
A rule $l \arrz r\ \constraint{\varphi}$ is \emph{compatible} with $R
\in \{ \succsim, \succ \}$ if for all substitutions $\gamma$ which
respect the rule we have: $l\gamma\ R\ r\gamma$.

\begin{theorem}
A set of dependency pairs $\P$ is chain-free if and only if there is
a reduction pair $(\succsim,\succ)$ and we can write $\P = \P_1 \cup
\P_2$ such that $\P_2$ is chain-free, and:
\begin{itemize}
\item all $\rho \in \P_1$ are compatible with $\succ$ 
  and
      all $\rho \in \P_2$ are compatible with $\succsim$;
\item either all $\rho \in \Rules$ are compatible with $\succsim$, \\
  or
  all $\rho \in \P$ have the form $l \arrz \afun(\aterm_1,\ldots,
  \aterm_i)\ \constraint{\varphi}$ with all $\aterm_i \in
  \Terms(\Sigmalogic,\LVar(\rho))$;
\item 
  $\afun(\vec{v}) \succsim w$ if $\afun$ is a calculation symbol,
  $v_1,\ldots,v_n$ are values and $w$ is the value of
  $\afun(\vec{v})$.
\end{itemize}
\end{theorem}

\noindent
Note that all rules must be compatible with $\succsim$, unless the
subterms of the right-hand sides in $\P$ 
can only be instantiated to ground logical terms;
in this
(reasonably common!) case, we can ignore the rules in the termination
argument.  This is a weak step in the direction of \emph{usable
rules}, a full treatment of which is beyond the scope of this short
paper.

For the reduction pair, we might for instance use the recursive path
ordering described in~\cite{kop:nis:13}.  Alternatively, we could
consider \emph{polynomial interpretations}:

\begin{theorem}
Given a mapping $\toint$ which assigns to each function symbol $\afun
: [\asort_1 \times \ldots \times \asort_n] \arrtype \bsort \in
\Sigmaterms \cup \Sigmalogic$ an $n$-ary polynomial over $\Z$, and
a valuation $\alpha$ which maps each variable\linebreak to an integer, every
term $\aterm$ corresponds to an integer
$\overline{\toint}_\alpha(\aterm)$.
Let $\aterm \succ \bterm$ if for all $\alpha$:
$\overline{\toint}_\alpha(\aterm) > \max(0,
\overline{\toint}_\alpha(\bterm))$,
and $\aterm \succsim \bterm$ if for all $\alpha$:
$\overline{\toint}_\alpha(\aterm) = \overline{\toint}_\alpha(\bterm)$.
Then $(\succsim,\succ)$ is a reduction pair.
\end{theorem}

Here, $\succsim$ is an equivalence relation.  Alternatively we might
base $\succsim$ on the $\geq$ relation in $\Z$, but then we must pose
an additional weak monotonicity requirement on $\toint$.

\begin{example}\label{ex:sum}
We consider an LCTRS over the integers, without overflow.  This
example uses bounded iteration, which is common in systems derived
from imperative programs:
\[
\quad\quad
\summ(\avar,\bvar) \arrz \nul\ \constraint{\avar > \bvar}
\quad\quad\quad
\summ(\avar,\bvar) \arrz \avar + \summ(\avar+\one,\bvar)\ 
  \constraint{\avar \leq \bvar}
\]
This system admits one dependency pair:
$\up{\summ}(\avar,\bvar) \arrz \up{\summ}(\avar+\one,\bvar)\ 
\constraint{\avar \leq \bvar}$.
Neither the dependency graph nor the value criterion can handle this
pair.  We \emph{can} orient it using polynomial interpretations, with
$\toint(\summ) = \lambda n m.m-n+1$; integer functions and integers
are interpreted as themselves.  Then $x \leq y \Rightarrow y - x + 1
> \max(0, y - (x+1) + 1)$ is valid, so the pair is compatible with
$\succ$ as required.

Thus, $\DP(\Rules)$ is chain-free if and only if $\emptyset$ is
chain-free, which is obviously the case!
\end{example}

\section{Related Work}

The most important related work is~\cite{fal:kap:09}, where a
constrained term rewriting formalism over the integers is introduced,
and methods are developed to prove termination similar to the ones
discussed here.  The major difference with the current work is that
the authors of~\cite{fal:kap:09} impose very strong type restrictions:
they consider only theory symbols (of sort $\mathtt{int}$) and defined
symbols (of sort $\mathtt{unit}$).  Rules have the form
$\afun(\avar_1,\ldots,\avar_n) \arrz \bfun(\aterm_1,\ldots,\aterm_n)$,
where the $\avar_i$ are variables and all $\aterm_i$ are logical
terms.  This significantly simplifies the analysis (for example, the
dependency pairs are exactly the rules), but has more limited
applications; it suffices for proving termination of simple
(imperative) integer programs, but does not help directly for
analysing confluence or function equivalence.

\section{Conclusion}

In this paper, we have seen how termination methods for normal TRSs,
and in particular the dependency pair approach, extend naturally to
the setting of LCTRSs.  Decision procedures are handled by solving
validity of logical formulas.  While this is undecidable in general,
many practical cases can be handled using today's powerful
SMT-solvers.

Considering termination results, we have only seen the tip of the
iceberg.  In the future, we hope to extend the constrained dependency
pair framework to handle also innermost termination and
non-termination.  Moreover, the dependency pair approach can be
strengthened with various techniques for simplifying dependency pair
processors, both adaptations of existing techniques for unconstrained
term rewriting (such as usable rules) and specific methods for
constrained term rewriting (such as the \emph{chaining} method used
in~\cite{fal:kap:09} or methods to add constraints in some cases).

In addition, we hope to provide an automated termination tool for
LCTRSs in the near future.  Such a tool could for instance be coupled
with a transformation tool from e.g.\ C or Java to be
immediately applicable for proving termination of imperative
programs, or can be used as a back-end for analysis tools of confluence
or function equivalence.

\bibliography{references}

\begin{thebibliography}{1}

\bibitem{art:gie:00}
T.~Arts and J.~Giesl.
\newblock Termination of term rewriting using dependency pairs.
\newblock {\em TCS}, 236(1-2):133--178, 2000.

\bibitem{fal:kap:09}
S.~Falke and D.~Kapur.
\newblock A term rewriting approach to the automated termination analysis of
  imperative programs.
\newblock In {\em Proc.\ CADE~09}, volume 5663 of {\em LNCS}, pages 277--293.
  Springer, 2009.

\bibitem{fur:nis:sak:kus:sak:08}
Y.~Furuichi, N.~Nishida, M.~Sakai, K.~Kusakari, and T.~Sakabe.
\newblock Approach to procedural-program verification based on implicit
  induction of constrained term rewriting systems.
\newblock {\em IPSJ Transactions on Programming}, 1(2):100--121, 2008.
\newblock In Japanese.

\bibitem{kop:nis:13}
C.~Kop and N.~Nishida.
\newblock Term rewriting with logical constraints.
\newblock In {\em Proc.\ FroCoS~13}, volume 8152 of {\em LNAI}, pages 343--358.
  Springer, 2013.

\bibitem{sak:nis:sak:11}
T.~Sakata, N.~Nishida, and T.~Sakabe.
\newblock On proving termination of constrained term rewrite systems by
  eliminating edges from dependency graphs.
\newblock In {\em Proc.\ WFLP~11}, LNCS, pages 138--155. Springer, 2011.

\bibitem{ter:03}
Terese.
\newblock {\em Term Rewriting Systems}, volume~55 of {\em Cambridge Tracts in
  TCS}.
\newblock Cambridge University Press, 2003.

\end{thebibliography}

\end{document}